\documentclass{amsart}
\usepackage[utf8]{inputenc}
\usepackage{amsthm,amsmath,amsfonts,amssymb}
\usepackage{adjustbox,graphicx,float}
\theoremstyle{definition}

\usepackage{soul,xcolor,comment}
\usepackage{hyperref}

\newtheorem{theorem}{Theorem}[section]
\newtheorem{corollary}{Corollary}[theorem]
\newtheorem{lemma}[theorem]{Lemma}

\newtheorem{condition}{Condition}
\theoremstyle{definition}
\newtheorem{definition}{Definition}[section]

\title{Additive Polycyclic Codes over $\mathbb{F}_{4}$ Induced by Binary Vectors and Some Optimal Codes}

\usepackage{lipsum}

\usepackage[foot]{amsaddr}

\author{Arezoo Soufi Karbaski}
\address{Department of Mathematics, Bu Ali Sina University, Hamedan, Iran}
\email{arezoo.sufi@basu.ac.ir}

\author{Taher Abualrub}
\address{Department of Mathematics \& Statistics,
American University of Sharjah, Sharjah, UAE }
\email{abualrub@aus.edu}

\author{Nuh Aydin}
\address{Department of Mathematics \& Statistics, Kenyon College, Gambier, OH, USA}
\email[]{aydinn@kenyon.edu}

\author{Peihan Liu}
\address{Department of Mathematics, University of Michigan, Ann Arbor, MI, USA}
\email[]{paulliu@umich.edu}

\begin{document}


\maketitle

\begin{abstract}
In this paper we  study the structure and properties of additive right
and left polycyclic codes induced by a binary vector $a$ in $\mathbb{F}%
_{2}^{n}.$ We  find the generator polynomials and the cardinality of
these codes. We  also study different duals for these codes. In
particular, we  show that if $C$ is a right polycyclic code induced by a
vector $a\in \mathbb{F}_{2}^{n}$, then the Hermitian dual of $C$ is a
sequential code induced by $a.$ As an application of these codes, we 
present examples of additive right polycyclic codes over $\mathbb{F}_{4}$
with more codewords than comparable optimal linear codes as well as optimal
binary linear codes and optimal quantum codes obtained from additive right
polycyclic codes over $\mathbb{F}_{4}.$
\end{abstract}

\section{Preliminaries}

A linear code of length $n$ over a finite field $\mathbb{F}$ is a subspace
of $\mathbb{F}^{n}.$ An Additive code of length $n$ over a finite field $%
\mathbb{F}$ is a subgroup of $\mathbb{F}^{n}.$ Additive codes over the
finite field $\mathbb{F}_{4}=\left\{ 0,1,\alpha ,\alpha ^{2}\right\} $ where 
$\alpha ^{2}+\alpha +1=0$ were introduced in \cite{Calderbank1998} because
of their applications in quantum computing. Define the mapping $T:\mathbb{F}%
_{4}^{n}\rightarrow \mathbb{F}_{4}^{n}$ by $T\left( \left(
c_{0},c_{1},\ldots ,c_{n-1}\right) \right) =\left(
c_{n-1},c_{0},c_{1},\ldots ,c_{n-2}\right) .$ Cyclic additive codes over $%
\mathbb{F}_{4}^{n}$ are additive codes over $\mathbb{F}_{4}^{n}$ such that
if $c=\left( c_{0},c_{1},\ldots ,c_{n-1}\right) \in C,$ then $T\left(
c\right) \in C.$

Let $C$ be a linear code over a finite field $\mathbb{F}.$ Right and left
linear polycyclic codes over finite fields were introduced in \cite%
{Sergio2009}. These codes are generalization of cyclic codes over finite
fields. Their properties and structures are studied in details in (\cite%
{Sergio2009},\cite{Matsuoka2012},\cite{Adel2016},\cite{Huffman2007}).

In this paper, we are interested in studying the structure and the properties
of additive right and left polycyclic codes induced by a binary vector $a=\left( a_{0},a_{1},\ldots ,a_{n-1}\right) \in \mathbb{F}_{2}^{n}.$ In this work, we consider the important special case where the  vector $a$ has all binary entries. Useful results are obtained in this special case both in terms of the structure of the codes and obtaining codes with good parameters.
We  give the definition of these codes, study their properties and find
their generator polynomials and their cardinality. We  also
study different duals of these codes and show that if $C$ is a right
polycyclic code induced by a vector $a\in \mathbb{F}_{2}^{n}$, then the
Hermitian dual of $C$ is a sequential code induced by $a.$ As an application
of our study, we present examples of codes with good parameters. We have
three sets of examples. One is a set of additive polycyclic codes that
contain more codewords (twice as many) than comparable optimal linear codes
over $\mathbb{F}_{4}$ with the same length and the minimum distance. Another
one is a set of best known binary linear codes, most of which are also
optimal, that are obtained from additive right polycyclic codes over $%
\mathbb{F}_{4}$ via certain maps. And the third is a set of optimal quantum
codes according to the database \cite{database} obtained from
additive polycyclic codes. 

\section{Introduction}

Consider the finite filed $\mathbb{F}_{4}=\left\{ 0,1,\alpha ,\alpha
^{2}\right\} $ where $\alpha ^{2}+\alpha +1=0.$ An additive code $C$ of
length $n$ over $\mathbb{F}_{4}$ is a subgroup of $(\mathbb{F}_{4}^{n},+).$

\begin{definition}
\label{Defn Right polycyclic}Let $C$ be an additive code over $\mathbb{F}_{4}
$ and let $a=\left( a_{0},a_{1},\ldots ,a_{n-1}\right) \in \mathbb{F}_{2}^{n}$.  $C$ is called
an additive right polycyclic code induced by $a$, if for any  $c=\left( c_{0},c_{1},\ldots ,c_{n-1}\right) \in C$, we have 
\[
\left( 0,c_{0},c_{1},\ldots ,c_{n-2}\right) +c_{n-1}\left(
a_{0},a_{1},\ldots ,a_{n-1}\right) \in C.
\]
\end{definition}

\begin{definition}
\label{Defn Left polycyclic}Let $C$ be an additive code over $\mathbb{F}_{4}
$ and let $a=\left( a_{0},a_{1},\ldots ,a_{n-1}\right) \in \mathbb{F}_{2}^{n}$.  $C$ is called
an additive left polycyclic code induced by $a$, if for any  $c=\left( c_{0},c_{1},\ldots ,c_{n-1}\right) \in C$, we have 
\[
\left( c_{1},c_{2},\ldots ,c_{n-1},0\right) +c_{0}\left( a_{0},a_{1},\ldots
,a_{n-1}\right) \in C.
\]
\end{definition}

Notice that if $a=\left( 1,0,0,\ldots ,0\right) ,$ then right polycyclic
codes induced by $a$ are just the familiar cyclic codes.

Suppose that $a=\left( a_{0},a_{1},\ldots ,a_{n-1}\right) \in \mathbb{F}%
_{2}^{n}.$ As in the case of additive cyclic codes over finite fields, it is
 useful to have polynomial representations of additive (right or
left) polycyclic codes.

Let $a=\left( a_{0},a_{1},\ldots ,a_{n-1}\right) \in \mathbb{F}_{2}^{n}.$
The vector $a$ can be represented as $a\left( x\right) =a_{0}+a_{1}x+\ldots
+a_{n-1}x^{n-1}\in \mathbb{F}_{2}\left[ x\right] $. Consider the ring $R_{n}=%
\mathbb{F}_{4}\left[ x\right] /\left\langle x^{n}-a\left( x\right)
\right\rangle$, which is an $\mathbb{F}_{2}\left[ x\right] $%
-module. Let $c\left( x\right) =c_{0}+c_{1}x+\ldots +c_{n-1}x^{n-1}\in 
\mathbb{F}_{4}\left[ x\right] /\left\langle x^{n}-a\left( x\right)
\right\rangle .$ Then,%
\begin{eqnarray*}
xc\left( x\right)  &=&c_{0}x+c_{1}x^{2}+\ldots +c_{n-2}x^{n-1}+c_{n-1}x^{n}
\\
&=&c_{0}x+c_{1}x^{2}+\ldots +c_{n-2}x^{n-1}+c_{n-1}\left(
a_{0}+a_{1}x+\ldots +a_{n-1}x^{n-1}\right)  \\
&=&c_{n-1}a_{0}+x\left( c_{0}+c_{n-1}a_{1}\right) +\ldots +x^{n-1}\left(
c_{n-2}+c_{n-1}a_{n-1}\right) .
\end{eqnarray*}%
The polynomial representation of $xc\left( x\right) $ is $\left(
c_{n-1}a_{0},c_{0}+c_{n-1}a_{1},\ldots ,c_{n-2}+c_{n-1}a_{n-1}\right)
=\left( 0,c_{0},c_{1},\ldots ,c_{n-2}\right) +c_{n-1}\left(
a_{0},a_{1},\ldots ,a_{n-1}\right) .$

Similarly, let $c=\left( c_{0},c_{1},\ldots ,c_{n-1}\right) \in \mathbb{F}%
_{4}^{n}$ be represented by the polynomial $c\left( x\right)
=c_{n-1}+c_{n-2}x+\cdots +c_{0}x^{n-1}\in \mathbb{F}_{4}\left[ x\right]
/\left\langle x^{n}-a\left( x\right) \right\rangle ,$ where $a=\left(
a_{0},a_{1},\ldots ,a_{n-1}\right) \in \mathbb{F}_{2}^{n}$ and $a\left(
x\right) =a_{n-1}+a_{n-2}x+\ldots +a_{0}x^{n-1}\in \mathbb{F}_{2}\left[ x%
\right] $. Then,%
\begin{eqnarray*}
xc\left( x\right)  &=&c_{n-1}x+c_{n-2}x^{2}+\ldots +c_{1}x^{n-1}+c_{0}x^{n}
\\
&=&c_{n-1}x+c_{n-2}x^{2}+\ldots +c_{1}x^{n-1}+c_{0}\left(
a_{n-1}+a_{n-2}x+\ldots +a_{0}x^{n-1}\right)  \\
&=&c_{0}a_{n-1}+x\left( c_{n-1}+c_{0}a_{n-2}\right) +\ldots +x^{n-1}\left(
c_{1}+c_{0}a_{0}\right) .
\end{eqnarray*}%
The polynomial representation of $xc\left( x\right) $ is $\left(
c_{1}+c_{0}a_{0},\ldots ,c_{n-1}+c_{0}a_{n-2},c_{0}a_{n-1}\right) =\left(
c_{1},c_{2},\ldots ,c_{n-1},0\right) +c_{0}\left( a_{0},a_{1},\ldots
,a_{n-1}\right) .$ Hence, we get the following lemmas.

\begin{lemma}
\label{submodule-right}$C$ is an additive right polycyclic code induced by $%
a $ if and only if $C$ is an $\mathbb{F}_{2}\left[ x\right] $-submodule of $%
R_{n}.$
\end{lemma}

\begin{lemma}
\label{submodule-left}$C$ is an additive left polycyclic code induced by $a$
if and only if $C$ is an $\mathbb{F}_{2}\left[ x\right] $-submodule of $%
R_{n}.$
\end{lemma}

Let $C$ be an additive right polycyclic code induced by $\mathbf{a}.$ Then $%
C $ is invariant under right multiplication by the square matrix 
\[
D=\left[ 
\begin{array}{ccccc}
0 & 1 & 0 & 0 & \ldots \\ 
0 & 0 & 1 & 0 & \ldots \\ 
\vdots &  &  & \ddots &  \\ 
0 & 0 & \ldots & 0 & 1 \\ 
a_{0} & a_{1} & a_{2} & \ldots & a_{n-1}%
\end{array}%
\right] . 
\]%
Similarly, an additive left polycyclic code induced by $\mathbf{d}=\left(
d_{0},d_{1},\ldots ,d_{n-1}\right) $ is invariant under right multiplication
by the square matrix 
\[
E=\left[ 
\begin{array}{ccccc}
d_{0} & d_{1} & d_{2} & \ldots & d_{n-1} \\ 
1 & 0 & 0 & \ldots & 0 \\ 
0 & 1 & 0 & \ldots & 0 \\ 
\vdots &  & \ddots & 0 & 1 \\ 
0 & 0 & \ldots & 1 & 0%
\end{array}%
\right] . 
\]

\begin{theorem}
\label{right-left1}Let $C$ be an additive right polycyclic code induced by $%
\mathbf{a}=\left( a_{0},a_{1},\ldots ,a_{n-1}\right) $ with $a_{0}\neq 0$.
Then $C$ is an additive left polycyclic code induced by $\mathbf{d}=\left(
d_{0},d_{1},\ldots ,d_{n-1}\right) $ where $d_{j}=\frac{-{a_{j+1}}}{a_{0}}$
for $j<n-1$ and $d_{n-1}=\frac{1}{a_{0}}$.
\end{theorem}

\begin{proof}
It is obvious that the matrix $D$ above is invertible and 
\[
D^{-1}=\left[ 
\begin{array}{ccccc}
d_{0} & d_{1} & d_{2} & \ldots & d_{n-1} \\ 
1 & 0 & 0 & \ldots & 0 \\ 
0 & 1 & 0 & \ldots & 0 \\ 
\vdots &  & \ddots & 0 &  \\ 
0 & 0 & \ldots & 1 & 0%
\end{array}%
\right] , 
\]%
where $d_{j}=\frac{-{a_{j+1}}}{a_{0}}$ for $j<n-1$ and $d_{n-1}=\frac{1}{%
a_{0}}$. Since $CD=C$, then $C=CDD^{-1}=CD^{-1}$ and hence $C$ is a left
polycyclic code induced by $\mathbf{d}=\left( d_{0},d_{1},\ldots
,d_{n-1}\right) $.
\end{proof}

Note that if $C$ is both additive right and left polycyclic code induced by
the same $\mathbf{a}$, it is an additive right-left polycyclic code.

\begin{theorem}
\label{left-right-2}Let $C$ be an additive right-left polycyclic code
induced by $a\left( x\right) =a_{0}+a_{1}x+\ldots +a_{n-1}x^{n-1}.$ Then $%
a\left( x\right) =\frac{x^{n}+a_{0}}{1+a_{0}x},$ where $a_{0}^{n+1}=\left(
-1\right) ^{n+1}.$
\end{theorem}

\begin{proof}
Since $C$ be an additive right-left polycyclic code induced by $a\left(
x\right) =a_{0}+a_{1}x+\ldots +a_{n-1}x^{n-1},$ then $a\left( x\right)
=d\left( x\right) ,$ where 
\begin{eqnarray*}
d\left( x\right) &=&\frac{-a_{1}}{a_{0}}+\frac{-a_{2}}{a_{0}}x+\frac{-a_{3}}{%
a_{0}}x^{2}+\ldots +\frac{-a_{n-1}}{a_{0}}x^{n-2}+\frac{1}{a_{0}}x^{n-1} \\
&=&-\frac{a_{0}+a_{1}x+\ldots +a_{n-1}x^{n-1}-a_{0}}{a_{0}x}+\frac{1}{a_{0}}%
x^{n-1} \\
&=&-\frac{a\left( x\right) -a_{0}}{a_{0}x}+\frac{1}{a_{0}}x^{n-1}.
\end{eqnarray*}%
Since $a\left( x\right) =d\left( x\right) ,$ then $a\left( x\right) =-\frac{%
a\left( x\right) -a_{0}}{a_{0}x}+\frac{1}{a_{0}}x^{n-1}=\frac{-a\left(
x\right) +a_{0}+x^{n}}{a_{0}x}.$ This implies that $a_{0}xa\left( x\right)
=-a\left( x\right) +a_{0}+x^{n}$ and $a\left( x\right) =\frac{x^{n}+a_{0}}{%
1+a_{0}x}.$ Since $a\left( x\right) =a_{0}+a_{1}x+\ldots +a_{n-1}x^{n-1},$
then the root of $1+a_{0}x$ must be also the root of $x^{n}+a_{0}.$
Therefore, $a_{0}^{n+1}=\left( -1\right) ^{n+1}.$
\end{proof}

\section{The generators of additive right and left polycyclic Codes}

In this section, we will focus on finding the structure of additive right
polycyclic codes over $\mathbb{F}_{4}$ induced by a binary vector $%
a=\left( a_{0},a_{1},\ldots ,a_{n-1}\right) \in \mathbb{F}_{2}^{n}.$ The
structure of additive left polycyclic codes over $\mathbb{F}_{4}$ is obtained in
the same way.

Suppose that $C$ is an additive right polycyclic code over $\mathbb{F}_{4}$
induced by $a=\left( a_{0},a_{1},\ldots ,a_{n-1}\right) \in \mathbb{F}%
_{2}^{n}.$ From Lemma \ref{submodule-right}, $C$ is an $\mathbb{F}_{2}\left[
x\right] $-submodule of $R_{n}.$ For the rest of the paper, we will simply write $f$ for the polynomial $f\left( x\right)$.

Let $\alpha g_{1}+g_{2},b$ be two polynomials in $C$ with the following
properties:

\begin{condition}
\label{condition 1} The polynomial $b$ is a binary polynomial of degree less
than $n$ in $C$ and of minimal degree. If $C$ has no binary polynomials,
then $b=0$
\end{condition}

\begin{condition}
\label{condition2}The polynomial $\alpha g_{1}+g_{2}$ is a nonbinary
polynomial of degree less than $n$ in $C$ in which $g_{1}$ and $g_{2}$ are
binary polynomials and $g_{1}$ is of minimal degree. If $C$ has only binary
polynomials, then $\alpha g_{1}+g_{2}=0$.
\end{condition}

Suppose that $c\in C.$ If $c$ is a binary polynomial in $C,$ then by the
division algorithm we get that%
\begin{eqnarray*}
c &=&qb+r\text{, or} \\
r &=&c-qb\text{ }
\end{eqnarray*}%
where $q$ and $r$ are binary polynomials and $\deg r<\deg c$ or $r=0.$
Hence, $r=c-qb\in C$ and $r=0.$ Therefore, $c=qb.$ Now suppose that $%
c=\alpha c_{1}+c_{2}$ is a nonbinary polynomial in $C.$ Then,%
\[
c_{1}=q_{1}g_{1}+r_{1}, 
\]%
where $q_{1}$ and $r_{1}$ are binary polynomials and $\deg r_{1}<\deg g_{1}.$
Hence,%
\begin{eqnarray*}
\alpha c_{1} &=&q_{1}\alpha g_{1}+\alpha r_{1}. \\
&=&q_{1}\left( \alpha g_{1}+g_{2}\right) +q_{1}g_{2}+\alpha r_{1}
\end{eqnarray*}%
Also,%
\[
c_{2}=q_{2}b+r_{2} 
\]%
where $q_{2},r_{2}$ are binary polynomials with $\deg r_{2}<\deg b.$ Hence,%
\begin{eqnarray*}
\alpha c_{1}+c_{2} &=&q_{1}\left( \alpha g_{1}+g_{2}\right)
+q_{1}g_{2}+\alpha r_{1}+q_{2}b+r_{2} \\
&=&q_{1}\left( \alpha g_{1}+g_{2}\right) +q_{2}b+\left( \alpha
r_{1}+q_{1}g_{2}+r_{2}\right) .
\end{eqnarray*}%
This implies that $\left( \alpha r_{1}+q_{1}g_{2}+r_{2}\right) \in C.$ Since 
$q_{1}g_{2}+r_{2}$ is a binary polynomial and $\deg r_{1}<\deg g_{1}$, 
we conclude that $r_{1}=0.$ Thus, $\left( q_{1}g_{2}+r_{2}\right) $ is a
binary polynomial in $C.$ Hence,\ $q_{1}g_{2}+r_{2}=q_{3}b$ and%
\begin{eqnarray*}
\alpha c_{1}+c_{2} &=&q_{1}\left( \alpha g_{1}+g_{2}\right) +q_{2}b+q_{3}b \\
&=&q_{1}\left( \alpha g_{1}+g_{2}\right) +\left( q_{2}+q_{3}\right) b
\end{eqnarray*}%
where $q_{1},q_{2},q_{3}$ are all binary polynomials. This proves the
following Theorem that classifies all right polycyclic codes.

\begin{theorem}
\label{main-right}Let $C$ be an additive right polycyclic code induced by $a. $ Then, $C=\left\langle \alpha g_{1}+g_{2},b\right\rangle $ where $\alpha
g_{1}+g_{2},b$ satisfy conditions \ref{condition 1} and \ref{condition2} and 
$\deg g_{2}<\deg b.$
\end{theorem}


\begin{proof}
We only need to prove that $\deg g_{2}<\deg b.$ Suppose that $\deg g_{2}\geq
\deg b.$ Since $g_{2}$ and $b$ are binary polynomials,  by the division
algorithm we have $g_{2}=qb+r,$ where $r=0$ or $\deg r<\deg b.$ Consider the
additive right polycyclic code $D=\left\langle \alpha g_{1}+r,b\right\rangle 
$ induced by $a.$ Then $\alpha g_{1}+r=\alpha g_{1}+g_{2}-qb\in C.$ Hence, $%
D\subseteq C.$ On the other hand, we have $\alpha g_{1}+g_{2}=\alpha
g_{1}+r+qb\in D.$ Thus, $C\subseteq D$ and therefore $C=D.$
\end{proof}

\begin{lemma}
\label{uniqueness-right}Suppose that $C=\left\langle \alpha
g_{1}+g_{2},b\right\rangle $ is an additive right polycyclic code induced by 
$a,$ with the generators satisfy the conditions in Theorem \ref{main-right}.
Then $g_{1},g_{2}$ and $b$ are unique.
\end{lemma}

\begin{proof}
Suppose that $C=\left\langle \alpha g_{1}+g_{2},b_{1}\right\rangle
=\left\langle \alpha g_{3}+g_{4},b_{2}\right\rangle .$ Since $b_{1}$ is a
binary polynomial of minimal degree and $b_{2}$ is binary polynomials of
minimal degree in $C,$ then $\deg b_{1}=\deg b_{2}.$ If $b_{1}\neq b_{2},$
then $b_{3}=\left( b_{1}-b_{2}\right) $ is a binary polynomial in $C$ of
degree less than the degree of $b_{1}.$ Hence, $b_{1}-b_{2}=0$ and $%
b_{1}=b_{2}.$ Since $g_{1}$ and $g_{3}$ are of minimal degree in $C,$ then $%
\deg g_{1}=\deg g_{3}.$ If $g_{1}\neq g_{2},$ then $\alpha \left(
g_{1}-g_{3}\right) +\left( g_{2}-g_{4}\right) \in C,$ and $\deg \left(
g_{1}-g_{2}\right) <\deg g_{1}.$ A contradiction unless $g_{1}=g_{3}.$ If $%
g_{2}\neq g_{4},$ then $\left( \alpha g_{1}+g_{2}\right) -\left( \alpha
g_{3}+g_{4}\right) =g_{2}-g_{4}$ is a binary polynomial in $C$ of degree
less than the degree of $b.$ A contradiction. Hence $g_{2}=g_{4}$ and the
generators are unique.
\end{proof}

\begin{theorem}
\label{main-binary}Let $C$ be an additive right polycyclic code induced by a
binary vector $a.$ Then the code $C$ is given by

\begin{enumerate}
\item $C=\left\langle b\right\rangle ,$ where $b$ is a binary polynomial of
minimal degree in $C$ and $b|\left( x^{n}-a\right) .$ Or

\item $C=\left\langle \alpha g_{1}+g_{2}\right\rangle $ where $g_{1},g_{2}$
satisfy condition \ref{condition2}. Moreover, $g_{1}|\left( x^{n}-a\right) ,$
and $\left( x^{n}-a\right) |\left( \dfrac{x^{n}-a}{g_{1}}g_{2}\right) .$ Or

\item $C=\left\langle \alpha g_{1}+g_{2},b\right\rangle ,$ where $%
g_{1},g_{2} $ and $b$ are polynomials satisfy conditions \ref{condition 1}
and \ref{condition2}. Moreover, $b|\left( x^{n}-a\right) ,$ $g_{1}|\left(
x^{n}-a\right) ,$ $\deg g_{2}<\deg b$ and $b|\left( \dfrac{x^{n}-a}{g_{1}}%
g_{2}\right) $
\end{enumerate}
\end{theorem}

\begin{proof}
From Theorem \ref{main-right}, we know that $C=\left\langle \alpha
g_{1}+g_{2},b\right\rangle ,$ where $\alpha g_{1}+g_{2},b$ satisfy
conditions \ref{condition 1} and \ref{condition2} and $\deg g_{2}<\deg b$.

\begin{enumerate}
\item If $C$ has only binary polynomials, then $\alpha g_{1}+g_{2}=0$ and $%
C=\left\langle b\right\rangle .$ Since $a$ is a binary polynomial,  by
the division algorithm we have%
\[
x^{n}-a=Q_{1}b+R_{1}. 
\]%
where $Q_{1}$ and $R_{1}$ are binary polynomials with $\deg R_{1}<\deg b$ or 
$R_{1}=0.$ Hence, $R_{1}=Q_{1}b\in C.$ Therefore, $R_{1}=0$ and $b|\left(
x^{n}-a\right) .$

\item Suppose that $C$ has no binary polynomials. Then $b=0$ and $%
C=\left\langle \alpha g_{1}+g_{2}\right\rangle $ where $g_{1},g_{2}$ satisfy
condition \ref{condition2}. We also have%
\[
x^{n}-a=Q_{2}g_{1}+R_{2}, 
\]%
where $Q_{2},R_{2}$ are binary polynomials and $\deg R_{2}<\deg g_{1}$.
Hence, 
\[
\alpha \left( x^{n}-a\right) =Q_{2}\left( \alpha g_{1}+g_{2}\right) +\alpha
R_{2}+Q_{2}g_{2}. 
\]%
This implies that $\left( \alpha R_{2}+Q_{2}g_{2}\right) \in C.$ Since $%
Q_{2}g_{2}$ is a binary polynomial and $\deg R_{2}<\deg g_{1},$ t $R_{2}=0 $ and $g_{1}|\left( x^{n}-a\right) .$ Notice that $\dfrac{x^{n}-a}{%
g_{1}}\left( \alpha g_{1}+g_{2}\right) =\dfrac{x^{n}-a}{g_{1}}g_{2}$ is a
binary polynomial in $C.$ Since $C$ has no binary polynomials, $\left(
x^{n}-a\right) |\left( \dfrac{x^{n}-a}{g_{1}}g_{2}\right) .$

\item If $C$ has binary and nonbinary polynomials, then $C=\left\langle
\alpha g_{1}+g_{2},b\right\rangle ,$ where $g_{1},g_{2}$ and $b$ are
polynomials satisfy conditions \ref{condition 1}, \ref{condition2} and $\deg
g_{2}<\deg b.$ From (1) and (2) we get that $b|\left( x^{n}-a\right) ,$ $%
g_{1}|\left( x^{n}-a\right) .$ Since $\dfrac{x^{n}-a}{g_{1}}\left( \alpha
g_{1}+g_{2}\right) =\dfrac{x^{n}-a}{g_{1}}g_{2}$ is a binary polynomial in $%
C,$ then $b|\left( \dfrac{x^{n}-a}{g_{1}}g_{2}\right) .$
\end{enumerate}
\end{proof}

\begin{corollary}
If $a\left( x\right) $ is a binary polynomial and $\gcd \left( b,\dfrac{%
x^{n}-a}{g_{1}}\right) =1,$ then $g_{2}=0.$
\end{corollary}

\begin{proof}
Since $b|\left( \dfrac{x^{n}-a}{g_{1}}g_{2}\right) $ and $\gcd \left( b,%
\dfrac{x^{n}-a}{g_{1}}\right) =1,$ then $b|g_{2}.$ But $\deg g_{2}<\deg b.$
Hence, $g_{2}=0.$
\end{proof}

\begin{corollary}
\label{main-left} Let $C$ be an additive left polycyclic code induced by a
binary vector $a.$ Then the code $C$ is given by

\label{Construction}
\begin{enumerate}
\item $C=\left\langle b\right\rangle ,$ where $b$ is a binary polynomial of
minimal degree in $C$ and $b|\left( x^{n}-a\right) .$ Or

\item $C=\left\langle \alpha g_{1}+g_{2}\right\rangle $ where $g_{1},g_{2}$
satisfy condition \ref{condition2}. Moreover, $g_{1}|\left( x^{n}-a\right) ,$
and $\left( x^{n}-a\right) |\left( \dfrac{x^{n}-a}{g_{1}}g_{2}\right) .$ Or

\item $C=\left\langle \alpha g_{1}+g_{2},b\right\rangle ,$ where $g_{1},g_{2}
$ and $b$ are polynomials satisfy conditions \ref{condition 1} and \ref%
{condition2}. Moreover, $b|\left( x^{n}-a\right) ,$ $g_{1}|\left(
x^{n}-a\right) ,$ $\deg g_{2}<\deg b$ and $b|\left( \dfrac{x^{n}-a}{g_{1}}%
g_{2}\right) $
\end{enumerate}
\end{corollary}

\begin{proof}
The proof is similar to the proof in Theorem \ref{main-right}, Lemma \ref%
{uniqueness-right} and Theorem \ref{main-binary}.
\end{proof}

\begin{theorem}
\label{Cardinality-binary}Let $C$ be an additive right polycyclic code
induced by a binary vector $a.$

\begin{enumerate}
\item If $C=\left\langle b\right\rangle ,$ where $b$ is a binary polynomial
with $\deg b=t_{1},$ then $\left\vert C\right\vert =2^{n-t_{1}}.$

\item If $C=\left\langle \alpha g_{1}+g_{2}\right\rangle ,$ where $\deg
g_{1}=t_{2}$ and $\left( x^{n}-a\right) |\left( \dfrac{x^{n}-a}{g_{1}}%
g_{2}\right) ,$ then $\left\vert C\right\vert =2^{n-\deg t_{2}}.$

\item If $C=\left\langle \alpha g_{1}+g_{2},b\right\rangle ,$ where $\deg
b=t_{1}$ and $\deg g_{1}=t_{2},$ then $\left\vert C\right\vert
=2^{n-t_{1}}2^{n-t_{2}}$
\end{enumerate}
\end{theorem}

\begin{proof}
We will prove (3). The proof of (1) and (2) is similar to the proof of (3).

Suppose that $C=\left\langle \alpha g_{1}+g_{2},b\right\rangle $ where $%
g_{1},g_{2},$ and $b$ satisfy the conditions in Theorem \ref{main-binary}.
Let $c\left( x\right) \in C.$ Then $c=f_{1}\left( \alpha g_{1}+g_{2}\right)
+f_{2}b.$ Using the division algorithm, we get 
\[
f_{2}=q_{2}\frac{x^{n}-a}{b}+r_{2}, 
\]%
where $r_{2}$ is a binary polynomial satisfying $\deg r_{2}<n-t_{1}.$ Hence, 
\[
f_{2}b=q_{2}\frac{x^{n}-a}{b}b+r_{2}b=r_{2}b, 
\]%
where $\deg r_{2}<n-t_{1}.$ Again using the division algorithm, we get 
\[
f_{1}=q_{1}\frac{x^{n}-a}{g_{1}}+r_{1}, 
\]%
where $r_{1}$ is a binary polynomial satisfying $\deg r_{1}<n-t_{2}.$ Hence, 
\begin{eqnarray*}
f_{1}\left( \alpha g_{1}+g_{2}\right) &=&q_{1}\frac{x^{n}-a}{g_{1}}\left(
\alpha g_{1}+g_{2}\right) +r_{1}\left( \alpha g_{1}+g_{2}\right) \\
&=&r_{1}\left( \alpha g_{1}+g_{2}\right) +q_{1}\frac{x^{n}-a}{g_{1}}g_{2}.
\end{eqnarray*}%
Since $q_{1}\dfrac{x^{n}-a}{g_{1}}g_{2}$ is a binary polynomial, $q_{1}%
\dfrac{x^{n}-a}{g_{1}}g_{2}=r_{3}b.$ We have 
\[
r_{3}=q_{3}\frac{x^{n}-a}{b}+r_{4}, 
\]%
where $\deg r_{4}<n-t_{1}.$ Hence, 
\begin{eqnarray*}
r_{3}b &=&q_{3}\frac{x^{n}-a}{b}b+r_{4}b \\
&=&r_{4}b.
\end{eqnarray*}%
Therefore, 
\begin{eqnarray*}
c &=&f_{1}\left( \alpha g_{1}+g_{2}\right) +f_{2}b \\
&=&r_{1}\left( \alpha g_{1}+g_{2}\right) +\left( r_{2}+r_{4}\right) b,
\end{eqnarray*}%
where $\deg r_{1}<n-t_{2}$ and $\deg \left( r_{2}+r_{3}\right) <n-t_{1}.$
Hence,\ $\left\vert C\right\vert =2^{n-t_{1}}2^{n-t_{2}}.$
\end{proof}

\section{Duals and Additive Sequential Codes}

\begin{definition}
Let $C$ be an additive code over $\mathbb{F}_{4}$ and let $c=\left(
c_{0},c_{1},\ldots ,c_{n-1}\right) \in C.$ $C$ is called an additive right
sequential code induced by $a$ if there exists a vector $a=\left(
a_{0},a_{1},\ldots ,a_{n-1}\right) \in \mathbb{F}_{2}^{n}$ such that%
\[
\left( c_{1},c_{2},\ldots ,c_{n-1},c_{0}a_{0}+c_{1}a_{1}+\ldots
+c_{n-1}a_{n-1}\right) \in C.
\]
\end{definition}

The trace map $Tr_{2}:\mathbb{F}_{4}\rightarrow \mathbb{F}_{2}$ is defined
by $Tr_{2}\left( a\right) =a+a^{2}.$ Thus, $Tr_{2}\left( 0\right)
=Tr_{2}\left( 1\right) =0,$ and $Tr_{2}\left( \alpha \right) =Tr_{2}\left(
\alpha ^{2}\right) =1.$ Moreover, For any $\mathbf{x},\mathbf{y}\in \mathbb{F%
}_{4}^{n},$ the trace inner product $\left\langle .,.\right\rangle _{T}$ on $%
\mathbb{F}_{4}$ is defined by%
\[
\left\langle \mathbf{x},\mathbf{y}\right\rangle _{T}=Tr_{2}\left(
\left\langle \mathbf{x},\mathbf{y}\right\rangle \right) ,
\]%
where $\left\langle \mathbf{x},\mathbf{y}\right\rangle $ is the Hermitian
inner product.

\begin{definition}
Let $C$ be an additive polycyclic code over $\mathbb{F}_{4}.$ We consider
three different dual codes for $C$

\begin{enumerate}
\item Define the Hermitian dual of $C$ by%
\[
C^{\perp }=\left\{ \mathbf{y}\in \mathbb{F}_{4}^{n}:\left\langle \mathbf{x},%
\mathbf{y}\right\rangle _{T}=0~\forall ~\mathbf{x\in C}\right\} . 
\]

\item Define the annihilator dual of $C$ by $Ann\left( C\right) =\left\{
g\in R_{n}:g\left( x\right) f\left( x\right) =0~\forall f\left( x\right) \in
C\right\} .$

\item Define the 0-dual of $C$ by%
\[
C^{0}=\left\{ g\in R_{n}:gf\left( 0\right) =0~\forall f\left( x\right) \in
C\right\} . 
\]
\end{enumerate}
\end{definition}

Since the set $R_{n}$ is an $\mathbb{F}_{4}\left[ x\right] $-submodule and $%
C $ is a subset of $R_{n},$  the sets $Ann\left( C\right) $ and $C^{0}$
above are well-defined. In [2], the authors proved that if $C$ is a linear
polycyclic code, then $Ann\left( C\right) =C^{0}.$

\begin{lemma}
\label{Dual1}Let $C$ be an additive right polycyclic code over $\mathbb{F}%
_{4}.$ Then $Ann\left( C\right) $ and $C^{0}$ are linear right polycyclic
codes over $\mathbb{F}_{4}.$
\end{lemma}

\begin{proof}
We will prove that $Ann\left( C\right) $ is an ideal in $R_{n}.$ The same
proof  applies to $C^{0}$ as well. Suppose that $s_{1},s_{2}\in Ann\left(
C\right) ,~f\in C$ and $r\in R_{n}.$ Then, $\left( s_{1}+s_{2}\right)
f=s_{1}f+s_{2}f=0$ and $r\left( x\right) s_{1}\left( x\right) f\left(
x\right) =r(x)\left( s_{1}\left( x\right) f\left( x\right) \right) =0.$
Hence, $s_{1}+s_{2}$ and $rs_{1}\in Ann\left( C\right) .$ Therefore $Ann(C)$
and $C^{0}$ are ideals in $R_{n}$ and hence they are linear right polycyclic
codes over $\mathbb{F}_{4}.$
\end{proof}

\begin{theorem}
\label{Dual2}$C$ is an additive right polycyclic code induced by $a$ if and
only if the dual code $C^{\perp }$ is an additive $\left( n,2^{2n-k}\right) $
sequential code induced by $a$.
\end{theorem}

\begin{proof}
Suppose that $C$ is an additive right polycyclic code induced by $a.$ Then $%
C $ is invariant under right multiplication by the square matrix 
\[
G=\left[ 
\begin{array}{ccccc}
0 & 1 & 0 & 0 & \ldots \\ 
0 & 0 & 1 & 0 & \ldots \\ 
\vdots &  &  & \ddots &  \\ 
0 & 0 & \ldots & 0 & 1 \\ 
a_{0} & a_{1} & a_{2} & \ldots & a_{n-1}%
\end{array}%
\right] . 
\]%
If $H$ is a check matrix for $C,$ then $\overline{C}H^{T}+C\overline{H}%
^{T}=C\ast H^{T}=0$, where $\ast $ means the trace inner product. Hence, $%
CD\ast H^{T}=\overline{CD}H^{T}+CD\overline{H}^{T}=C\left( H\overline{D}%
^{T}\right) ^{T}+C\left( \overline{H}D^{T}\right) ^{T}=C\ast \left(
HD^{T}\right) ^{T}=0$. This implies that $C^{\perp }$ is invariant under
right multiplication by $D^{T}.$ Thus $C^{\perp }$ is an additive sequential
code induced by $a$.
\end{proof}

\section{Applications and Examples}

\subsection{Additive Codes with More Codewords than  Optimal Linear Codes over GF(4)}
Based on Theorem \ref{main-binary} and Theorem \ref{Cardinality-binary} 
, we conducted computer searches to find additive codes with good parameters. We took binary $g_1|x^n-a$ and binary $b|x^n-a$ that are relatively prime, and binary $g_2$ to be a random polynomial of degree less than the degree of $b.$ Then, it automatically follows that $b|\frac{x^n-a}{g_1}g_2$. As a result of our search we found 13 additive polycyclic codes over $\mathbb{F}_{4}$ that contain more codewords than the comparable optimal linear codes $\mathbb{F}_{4}$. 

Table 1 below shows these codes with their generators  $\langle \alpha g_1 +g_2,b\rangle$. We use the notation $[n,\frac{2k+1}{2},d]_4$ to denote the parameters of a polycyclic code which means its size is $2\cdot 4^k$. On the other hand, the comparable optimal linear code with the same length and minimum distance $d$ has $4^k$ codewords and the best  linear code of dimension $k+1$ has minimum distance $d-1$. In other words, these additive polycyclic codes contain twice as many codewords as the optimal linear codes with the same length and minimum distance. Note that the multinomials below refer to $x^n-a$.

\begin{table}[h]
\caption{Additive codes $[n,k,d]_4$ v.s. BKLC $[n,k,d]_4$ with smaller dimension}
\label{tab-1}
\resizebox{.89\textwidth}{!}{%

\begin{tabular}{p{1.8cm} p{1.8cm} p{2.5cm} p{5cm} p{5cm} }
\hline\noalign{\smallskip}
$[n,\frac{2k+1}{2},d]_4$ &  $[n,k,d]_4$  &  $[n,k+1,d-1]_4$ & $\langle \alpha g_1 +g_2,b\rangle$ & Multinomial   \\
\hline\noalign{\smallskip}
$[ 7, 9/2, 3]_4$ & $[ 7, 4, 3 ]_4$ & $[ 7, 5, 2 ]_4$ &
$\langle\alpha (x^2 + x + 1)+x,x^3 + x^2 + 1\rangle$ & $x^7 + x^6 + x^5 + x^3 + 1$\\

$[7, 7/2, 4]_4$ & $[7, 3, 4]_4$ & $[7, 4, 3]_4$ &
$ \langle\alpha (x + 1)+ x^4 + x^3 + x^2 + x,    x^6 + x^5 + x^4 + x^3 + x^2 + x + 1\rangle$ & $x^7 + 1$\\

$[22, 37/2, 3]_4$ & $[22, 18, 3]_4$ & $[22, 19, 2]_4$ &
$ \langle\alpha (x + 1)+ x^3 + x, x^6 + x^4 + x^3 + x + 1\rangle$ & $x^{22} + x^{19} + x^{15} + x^{14} + x^{13} + x^8 + x^7 + x^6 + x^4 + x^2 + x + 1$\\

$[23, 39/2, 3]_4$ & $[23, 19, 3]_4$ & $[23, 20, 2]_4$ &
$ \langle\alpha (x + 1)+ x^4 + x^2,  x^6 + x^5 + 1\rangle$ & $x^{23} + x^{22} + x^{21} + x^{15} + x^{13} + x^{11} + x^{10} + x^9 + x^8 + x^7 + x^3 + x^2+ x + 1$\\

$[24, 41/2 , 3]_4$ & $[24, 20, 3]_4$ & $[24,21, 2]_4$ &
$ \langle\alpha (x^2 + x + 1)+ x^2 + x, x^5 + x^4 + x^3 + x + 1\rangle$ & $ x^{24} + x^{21} + x^{20} + x^{19} + x^{18} + x^{16} + x^{14} + x^8 + x^5 + x^4 + x^3 + x^2 1,$\\

$[25, 43/2, 4]_4$ & $[25, 21, 3]_4$ & $[25, 22, 2]_4$ &
$ \langle\alpha (x + 1)+x^4 + x^2, x^6 + x^5 + 1\rangle$ & $x^{25} + x^{24} + x^{22} + x^{21} + x^{19} + x^{18} + x^{15} + x^{13} + x^{12} + x^{11} + x^{10} +x^9 + x^8 + x^7 + x^6 + x^5 + x^4 + x^3 + x^2 + 1$\\

$[26, 45/2, 3]_4$ & $[26, 22, 3]_4$ & $[26, 23, 2]_4$ &
$ \langle\alpha (x + 1)+x^4 + x^3 + x,x^6 + x^5 + x^3 + x^2 + 1\rangle$ & $x^{26} + x^{24} + x^{21} + x^{16} + x^{15} + x^{13} + x^{12} + x^{11} + x^{10} + x^9 + x^7 +x^3 + x^2 + 1$\\

$[27, 47/2, 3]_4$ & $[27, 23, 3]_4$ & $[27, 24, 2]_4$ &
$ \langle\alpha (x^2 + x + 1)+ x^2 + x, x^5 + x^4 + x^3 + x + 1\rangle$ & $x^{27} + x^{26} + x^{23} + x^{21} + x^{19} + x^{18} + x^{17} + x^{14} + x^{13} + x^{11} + x^{10} +x^7 + x^3 + x + 1$\\

$[27, 47/2, 3]_4$ & $[27, 23, 3]_4$ & $[27, 24, 2]_4$ &
$ \langle\alpha (x^2 + x + 1)+ x^2 + x, x^5 + x^4 + x^3 + x + 1\rangle$ & $x^{27} + x^{26} + x^{23} + x^{21} + x^{19} + x^{18} + x^{17} + x^{14} + x^{13} + x^{11} + x^{10} +x^7 + x^3 + x + 1$\\

$[28, 49/2, 3,]_4$ & $[28, 24, 3]_4$ & $[28, 25, 2]_4$ &
$ \langle\alpha (x^2 + x + 1)+ x^3 + x^2 + x, x^5 + x^3 + x^2 + x + 1\rangle$ & $x^{28} + x^{27} + x^{24} + x^{23} + x^{22} + x^{20} + x^{17} + x^{16} + x^{15} + x^{14} + x^{13} +x^{11} + x^{10} + x^8 + x^7 + x^5 + x^3 + x + 1$\\

$[29, 51/2, 3]_4$ & $[29, 25, 3]_4$ & $[29, 26, 2]_4$ &
$ \langle\alpha (x + 1)+ x^2 + x + 1, x^6 + x^5 + x^4 + x + 1\rangle$ & $x^{29} + x^{25} + x^{22} + x^{21} + x^{19} + x^{15} + x^{14} + x^{13} + x^{11} + x^{10} + x^6 +x^3 + x + 1$\\

$[30, 53/2, 3]_4$ & $[30, 26, 3]_4$ & $[30, 27, 2]_4$ &
$ \langle\alpha (x + 1)+ x^3 + x^2 + 1, x^6 + x^5 + x^3 + x^2 + 1\rangle$ & $x^{30} + x^{27} + x^{26} + x^{25} + x^{24} + x^{20} + x^{19} + x^{18} + x^{16} + x^{11} + x^{10} +x^9 + x^8 + x^7 + x^2 + 1$\\

$[31, 55/2, 3]_4$ & $[31, 27, 3]_4$ & $[31, 28, 2]_4$ &
$ \langle\alpha ( x + 1)+ x^4 + x + 1, x^6 + x^5 + x^3 + x^2 + 1\rangle$ & $x^{31} + x^{25} + x^{23} + x^{22} + x^{20} + x^{18} + x^{17} + x^{16} + x^{15} + x^{14} + x^{13} +x^{12} + x^{10} + x^9 + x^8 + x^7 + x^6 + x^4 + x^3 +1$\\

\noalign{\smallskip}\hline
\end{tabular}%
}
\end{table}

\subsection{Best Known Linear Codes over GF(2) Obtained from Additive Codes over GF(4)}


We can define the following maps, $W$, $T$, and $L$: \cite{Abualrub2020} 
\begin{align*}
    W \quad &: \quad \mathbb{F}_{4}^{n} \quad \xrightarrow[]{} \quad \mathbb{F}_{2}^{2n}\\
    & W(x_1,x_2,...,x_n)\\
    = &W \big( (a_1+b_1\alpha),(a_2+b_2\alpha),...,(a_n+b_n\alpha)\big)\\
    = &\big( (a_1+b_1),(a_2+b_2),...,(a_n+b_n),b_1,b_2,...,b_n \big)\\
    \mbox{(trace map)}\quad T \quad &: \quad \mathbb{F}_{4}^{n} \quad \xrightarrow[]{} \quad \mathbb{F}_{2}^{n}\\
    & T(x_1,x_2,...,x_n)\\
    = & \big( (x_1+x_1^2),(x_2+x_2^2),...,(x_n+x_n^2)\big)\\
    L \quad &: \quad \mathbb{F}_{4}^{n} \quad \xrightarrow[]{} \quad \mathbb{F}_{2}^{n}\\
    & L(x_1,x_2,...,x_n)\\
    = & \big( (x_1\alpha+x_1^2\alpha^2),(x_2\alpha+x_2^2\alpha^2),...,(x_n\alpha+x_n^2\alpha^2)\big)
\end{align*}

Using these map, we can construct binary linear codes from quaternary additive codes. Namely, by $W$, we can construct an $[2n,2k,d']_2$ linear code from an $[n,k,d]_4$ additive code; and by $T$ and $L$, we can construct $[n,k',d']_2$ linear codes from $[n,k,d]_4$ additive codes. We present some optimal codes in Table 2 where  $[n,k,d]_2$ is the binary optimal linear codes we obtained, and  $\langle \alpha g_1 +g_2,b\rangle$ refers to the generators  as in Theorem \ref{main-binary}
, multinomial denotes $x^n-a$, where $a$ is the associated vector; $*$ and $\circ$ delineates reversible and self-orthogonal codes respectively.

\begin{table}[h]
\caption{Optimal binary  linear codes $[n,k,d]_2$ obtained from quaternary additive codes based on $W$, $T$, and $L$}
\label{tab-2}
\resizebox{.8\textwidth}{!}{%

\begin{tabular}{p{1.8cm} p{1cm} p{5cm} p{5cm} }
\hline\noalign{\smallskip}
$[n,k,d]_2$  & Map & $\langle \alpha g_1 +g_2,b\rangle$ & Multinomial   \\
\hline\noalign{\smallskip}
$[ 7, 2, 4]_2^{*\circ}$ & L &
$\langle\alpha (x^5 + x^4 + x + 1) + 1,x^2 + x + 1
\rangle$ & $x^7 + x^4 + x^3 + 1$\\

$[ 10, 4, 4]_2$ & L &
$\langle\alpha (x + 1)+x + 1,x^6 + x^5 + x^4 + x + 1\rangle$ & $x^{10} + x^9 + x^8 + x^6 + x + 1$\\

$[ 12, 5, 4]_2$ & L &
$\langle\alpha (x^2 + x + 1)+x^2 + x + 1,x^7 + x^3 + x^2 + x + 1\rangle$ & $x^{12} + x^6 + x^5 + x^4 + 1$\\

$[ 16, 9, 4]_2^*$ & T &
$\langle\alpha (x^7 + x^3 + x^2 + x + 1)+x^4 + x^3 + x^2 + x + 1,x^9 + x + 1
\rangle$ & $x^{16} + x^{12} + x^{11} + x^{10} + x^9 + x^8 + x^7 + x^4 + 1$\\

$[ 17, 9, 5]_2$ & T &
$\langle\alpha (x^8 + x^5 + x^4 + x^3 + 1)+x^7 + x^6 + x^4 + x^3 + x,x^9 + x^8 + x^6 + x^5 + x^4 + x^3 + x^2 + x + 1\rangle$ & $x^{17} + x^{16} + x^{13} + x^{12} + x^{11} + x^{10} + x^9 + x^8 + x^7 + x^4 + x^2 + x +  1$\\

$[ 20, 11, 5]_2$ & W &
$\langle\alpha (x+1)+x^6 + x^4 + x^2 + x,x^8 + x^6 + x^5 + x^4 + x^2 + x + 1\rangle$ & $x^{10} + x^7 + x^5 + x^3 + x + 1$\\

$[ 26, 17, 4]_2$ & W &
$\langle\alpha (x+1)+x^4 + x^3 + x^2,x^8 + x^5 + x^3 + x + 1\rangle$ & $x^{13} + x^{10} + x^6 + x^3 + x + 1$\\

$[ 35, 24, 5]_2^*$ & T &
$\langle\alpha (x^{11} + x^9 + x^8 + x^6 + x^4 + x^3 + x^2 + x + 1)+x^{22} + x^{16} + x^{15} + x^{14} + x^{13} + x^{10} + x^9 + x^8 + x^7 + x^5 + x + 1 , x^{24} + x^{21} + x^{19} + x^{16} + x^{13} + x^{12} + x^{10} + x^8 + x^7 + x^6 + x^2 + x +1\rangle$ & $x^{35} + x^{33} + x^{30} + x^{29} + x^{26} + x^{23} + x^{21} + x^{20} + x^{19} + x^{14} + x^{13} +x^{12} + x^{11} + x^9 + x^8 + x^7 + x^6 + x^4 + x^3 + x^2 + 1
$\\

$[ 49, 39, 4]_2^*$ & T &
$\langle\alpha ( x^{10} + x^8 + x^6 + x^4 + x^2 + x + 1)+x , x^5 + x^4 + x^3 + x^2 + 1
\rangle$ & $x^{49} + x^{46} + x^{45} + x^{44} + x^{40} + x^{39} + x^{38} + x^{36} + x^{35} + x^{31} + x^{30} +
x^{29} + x^{28} + x^{27} + x^{24} + x^{23} + x^{20} + x^{19} + x^{17} + x^{15} + x^{12} + x^{11} + x^9 + x^8 + x^4 + x + 1
$\\

$[ 62, 51, 4]_2$ & W &
$\langle\alpha ( x + 1)+x^8 + x^4 + x + 1 , x^10 + x^9 + x^8 + x^4 + 1
\rangle$ & $x^{31} + x^{29} + x^{28} + x^{27} + x^{26} + x^{25} + x^{23} + x^{15} + x^{14} + x^{13} + x^{12} +x^9 + x^8 + x^7 + x^3 + x^2 + x + 1
$\\

$[ 98, 86, 4]_2$ & W &
$\langle\alpha (x^2+ x + 1)+x^8 + x^7 + x^6 + x^3 + 1 , x^{10} + x^8 + x^6 + x^4 + x^2 + x + 1\rangle$ & $x^{49} + x^{48} + x^{46} + x^{45} + x^{43} + x^{39} + x^{36} + x^{31} + x^{27} + x^{25} + x^{24} +x^{22} + x^{21} + x^{20} + x^{19} + x^{18} + x^{17} + x^{14} + x^7 + x^6 + x^3 + x +  1$\\

\noalign{\smallskip}\hline
\end{tabular}%
}
\end{table}

\subsection{Quantum Codes with Good Parameters from Additive Polycyclic Codes}
Quantum error correcting codes are an increasingly important part of quantum information theory. Many researchers have used classical block codes to construct quantum codes. One of the most commonly used  methods is the CSS construction \cite{Calderbank1998}. The CSS construction requires two linear codes $C_1$ and $C_2$ such that $C_2^{\perp}\subseteq C_1$. Hence, if $C_1$ is a self-dual code, then we can construct a CSS quantum code using $C_1$ alone since $C_1^{\perp}\subseteq C_1$. If  $C_1$ is self-orthogonal or dual-containing, then we can construct a CSS quantum code with $C_1^{\perp}$ and $C_1$ with the similar reason. 

In this section, we  present some optimal quantum codes derived from additive polycyclic codes. After constructing additive polycyclic codes by theorem 12, we used maps W, T, and L to construct self-dua, self-orthogonal, or dual-containing binary codes. Then, we used binary codes and their duals to construct CSS codes. Here are some example of optimal codes we found using this method. The optimality of these codes can be confirmed via the online table \cite{database}.

\begin{table}[H]
\caption{Optimal quantum codes $[[n,k,d]]_4$ from self-dual/self-orthogonal/dual-containing binary linear codes obtained from quaternary additive codes}

\label{tab-3}
\resizebox{.8\textwidth}{!}{%

\begin{tabular}{p{1.8cm} p{1cm} p{6cm} p{5cm} }
\hline\noalign{\smallskip}
$[[n,k,d]]_4$  & Map & $\langle \alpha g_1 +g_2,b\rangle$ & Multinomial   \\
\hline\noalign{\smallskip}

$[[ 7,1,3]]_4$ & T &
$\langle\alpha (x^3+x^2 + 1)+ x,x^4 + x^3 + x^2 + x + 1
\rangle$ & $x^7 + x^4 + x^3 + x + 1 $\\

$[[ 10, 8, 2]]_4$ & T &
$\langle\alpha (x + 1)+ 1,x^2 + x + 1\rangle$ & $x^{10} + x^9 + x^7 + x^5 + x^4 + x^2 + x + 1 $\\

$[[ 15, 7, 3]]_4$ & T &
$\langle\alpha (x^4 + x + 1)+ x^2 + x + 1,x^4 + x^3 + 1
\rangle$ & $x^{15} + x^{10} + x^9 + x^7 + x^6 + x^5 + x^3 + x^2 + 1 $\\

$[[ 30, 26, 2]]_4$ & T &
$\langle\alpha (x^2 + x + 1)+ x^8 + x^6 + 1,x^{12} + x^{11} + x^8 + x^7 + x^5 + x^4 + x^2 + x + 1
\rangle$ & $x^{30} + x^{29} + x^{27} + x^{25} + x^{24} + x^{23} + x^{22} + x^{21} + x^{20} + x^{17} + x^{13} +
x^{11} + x^{10} + x^9 + x^6 + x^5 + x^2 + x + 1 $\\

$[[ 35, 29, 2]]_4$ & L &
$\langle\alpha ( x + 1)+ x + 1,x^3 + x + 1
\rangle$ & $x^{35} + x^{34} + x^{33} + x^{31} + x^{27} + x^{25} + x^{24} + x^{22} + x^{20} + x^{18} + x^{17} +
x^{16} + x^{15} + x^{11} + x^5 + x^4 + x + 1$\\

$[[ 36, 30, 2]]_4$ & W &
$\langle\alpha ( x + 1)+  1,x^2 + x + 1
\rangle$ & $x^{18} + x^{16} + x^{13} + x^{11} + x^9 + x^8 + x^7 + x^4 + x^3 + 1 $\\

$[[ 48, 42, 2]]_4$ & W &
$\langle\alpha ( x + 1)+  1,x^2 + x + 1
\rangle$ & $x^{24} + x^{22} + x^{16} + x^{15} + x^{14} + x^{13} + x^{11} + x^{10} + x^7 + x^6 + x^4 + 1 $\\

$[[ 56, 50, 2]]_4$ & L &
$\langle\alpha ( x + 1)+  1,x^3 + x + 1
\rangle$ & $x^{56} + x^{51} + x^{49} + x^{48} + x^{47} + x^{41} + x^{38} + x^{36} + x^{31} + x^{30} + x^{29} +
x^{28} + x^{26} + x^{18} + x^{16} + x^{15} + x^{14} + x^{12} + x^{10} + x^9 + x^5 + x^2 + x + 1 $\\

$[[ 72, 66, 2]]_4$ & W &
$\langle\alpha ( x + 1)+  1,x^2 + x + 1
\rangle$ & $x^{36} + x^{35} + x^{32} + x^{31} + x^{30} + x^{29} + x^{26} + x^{22} + x^{21} + x^{20} + x^{18} +x^{16} + x^{13} + x^{10} + x^9 + x^2 + x + 1 $\\

$[[ 84, 76, 2]]_4$ & W &
$\langle\alpha ( x + 1)+  1,x^3 + x + 1
\rangle$ & $x^{42} + x^{41} + x^{40} + x^{39} + x^{37} + x^{35} + x^{33} + x^{30} + x^{29} + x^{28} + x^{27} +x^{25} + x^{21} + x^{15} + x^{14} + x^{13} + x^{11} + x^{10} + x^9 + x^6 + x^5 + x^2 +
        x + 1 $\\
 
\noalign{\smallskip}\hline
\end{tabular}%
}
\end{table}
\section{Conclusion and Future Works}

In this paper, we have studied the structure and properties of additive
right and left polycyclic codes induced by a binary vector $a\in \mathbb{F}%
_{2}^{n}.$ We have also studied different duals of these codes and showed
that $C$ is an additive right polycyclic code induced by $a$ if and only if
the Hermitian dual code $C^{\perp }$ is an additive $\left(
n,2^{2n-k}\right) $ sequential code induced by $a$. Finally, we have shown that it is possible to obtain both classical and quantum codes with good
parameters from these
codes.

Our work in this paper is on additive right and left
polycyclic codes induced by a binary vector $a\in \mathbb{F}_{2}^{n}.$ A
generalization of this work to non-binary vectors $a\in \mathbb{F}_{4}^{n}$
would be interesting. It is known  that (\cite{Calderbank1998}) constacyclic
additive codes that are not cyclic over $\mathbb{F}_{4}$ are just linear
constacyclic codes. Hence, it will be interesting to study the relationship
between additive right and left polycyclic codes over $\mathbb{F}_{4}$
induced by a nonbinary vector $a\in \mathbb{F}_{4}^{n}$ and linear
polycyclic codes over $\mathbb{F}_{4}$ induced by the vector $a.$ We have
few results on this relationship but the work is still in progress.

\end{document}